\UseRawInputEncoding 

\documentclass[10pt]{article}
\usepackage[letterpaper, left=1in, top=1in, right=1in, bottom=1in, verbose, ignoremp]{geometry}

\usepackage{url}\RequirePackage[colorlinks,citecolor=blue, linkcolor=blue,urlcolor = blue]{hyperref}
\usepackage{latexsym,amssymb,amsmath,amsfonts,graphicx,color,fancyvrb,amsthm,enumerate,subcaption,mathrsfs}
\usepackage[longnamesfirst,authoryear,round]{natbib}
\usepackage[dvipsnames]{xcolor}
\usepackage{float}
\usepackage{bm}
\usepackage{bbm}
\usepackage{multicol,multirow}
\usepackage{array}
\usepackage{relsize}
\usepackage{chngcntr}
\usepackage{wrapfig}
\usepackage{etoolbox}
\usepackage{caption}
\usepackage{tikz,pgf,pgfplots}
\usetikzlibrary{pgfplots.groupplots}
\usetikzlibrary{decorations.pathreplacing,calc}
\usetikzlibrary{shapes,backgrounds}
\usetikzlibrary{patterns}
\usetikzlibrary{cd}
\usepackage{amsopn}
\usepackage{tabularx}
\usepackage{calc}
\usepackage{hyperref}
\usepackage{mathtools}

\usepackage{tikz-cd} 
\usepackage{amscd} 
\usepackage{comment}
\usepackage[capitalize,nameinlink,compress]{cleveref}
\crefname{algocf}{Algorithm}{Algorithms}
\Crefname{algocf}{Algorithm}{Algorithms}
\crefname{figure}{Figure}{Figures} 
\crefname{assumption}{Assumption}{Assumptions}
\crefname{subsection}{Subsection}{Subsections}
\usepackage[linesnumbered,ruled,vlined]{algorithm2e}
\usepackage[noend]{algpseudocode}

\newcounter{cdrow}

\newtheorem{theorem}{Theorem}[]
\newtheorem*{theorem*}{Theorem}

\newtheorem{prop}[theorem]{Proposition}

\newtheorem*{claim*}{Claim}

\theoremstyle{definition}
\newtheorem{definition}[theorem]{Definition}
\newtheorem*{definition*}{Definition}

\theoremstyle{remark}

\newtheorem{example}[theorem]{Example}
\newtheorem*{example*}{Example}

\makeatletter
\newcommand*{\op}{%
  \DOTSB
  \mathop{\vphantom{\bigoplus}\mathpalette\matt@op\relax}%
  \slimits@
}
\newcommand\matt@op[2]{%
  \vcenter{\m@th\hbox{\resizebox{\widthof{$#1\bigoplus$}}{!}{$\boxplus$}}}%
}
\makeatother

\newcommand{\IR}{\mathbb R}
\newcommand{\IN}{\mathbb N}

\newcommand{\Ff}{\mathcal F}

\definecolor{forest}{RGB}{30, 130, 30}
\definecolor{skyblue}{RGB}{0, 180, 250}

\SetKwInput{KwInput}{Input}               
\SetKwInput{KwOutput}{Output}

\makeatletter
\def\@biblabel#1{}
\makeatother

\makeatletter
\patchcmd{\NAT@citex}
  {\@citea\NAT@hyper@{%
     \NAT@nmfmt{\NAT@nm}%
     \hyper@natlinkbreak{\NAT@aysep\NAT@spacechar}{\@citeb\@extra@b@citeb}%
     \NAT@date}}
  {\@citea\NAT@nmfmt{\NAT@nm}%
   \NAT@aysep\NAT@spacechar\NAT@hyper@{\NAT@date}}{}{}

\patchcmd{\NAT@citex}
  {\@citea\NAT@hyper@{%
     \NAT@nmfmt{\NAT@nm}%
     \hyper@natlinkbreak{\NAT@spacechar\NAT@@open\if*#1*\else#1\NAT@spacechar\fi}%
       {\@citeb\@extra@b@citeb}%
     \NAT@date}}
  {\@citea\NAT@nmfmt{\NAT@nm}%
   \NAT@spacechar\NAT@@open\if*#1*\else#1\NAT@spacechar\fi\NAT@hyper@{\NAT@date}}
  {}{}

\makeatother

\begin{document}
\def\spacingset#1{\renewcommand{\baselinestretch}%
{#1}\small\normalsize} \spacingset{1}

\begin{flushleft}
{\Large{\textbf{Topological Community Detection: A Sheaf-Theoretic Approach}}}
\newline
\\
Arne Wolf$^{1,2,\dagger}$ and Anthea Monod$^{1,\dagger}$
\\
\bigskip
\bf{1} Department of Mathematics, Imperial College London, UK
\\
\bf{2} London School of Geometry and Number Theory, UK
\\
\bigskip
$\dagger$ Corresponding e-mails: a.wolf22@imperial.ac.uk, a.monod@imperial.ac.uk
\end{flushleft}


\section*{Abstract}
We propose a model for network community detection using topological data analysis,
a branch of modern data science that leverages theory from algebraic topology to statistical analysis and machine learning. Specifically, we use cellular sheaves, which relate local to global properties of various algebraic topological constructions, to propose three new algorithms for vertex clustering over networks to detect communities. We apply our algorithms to real social network data in numerical experiments and obtain near optimal results in terms of modularity. Our work is the first implementation of sheaves on real social network data and provides a solid proof-of-concept for future work using sheaves as tools to study complex systems captured by networks and simplicial complexes.

\paragraph{Keywords:} Cellular sheaves; community detection; modularity; opinion dynamics; topological data ana-lysis.
\\


\section{Introduction}
\label{sec:intro}

Networks are used to describe, study, and understand complex systems in many scientific disciplines.  One of the most important features in complex systems that networks are able to capture is the presence of \emph{communities}. In networks, communities can be seen as partitioning a graph into clusters, which are subsets of vertices with many edges connecting the vertices within the subset, and comparatively fewer edges connecting to different subsets in the rest of the network. These clusters or communities can be considered as relatively independent components of a graph.
The problem of \emph{community detection} is to locate those clusters in networks which are more strongly connected than the whole network is, on average.
Community detection is a challenging and active area of research in network science: a key aspect that makes the problem difficult is that there is no single, universally accepted definition of a community within a network and it is largely dependent on context or the specific system being studied.

Topological data analysis (TDA) is a recently-emerged approach to data science that uses principles from pure mathematics to extract meaningful information from large and complex datasets that may not possess a rigorous metric or vector space structure which is often required in classical data analysis. In our work, we focus on \emph{sheaves}, which are a tool that relates local to global properties of various constructions in algebraic topology and have been used to reinterpret and generalize many significant results in classical geometry and algebraic topology. Sheaves allow for information to be assigned to subsets of topological spaces; when the topological space is a network and the information is vector space-valued, a computational framework for sheaves becomes available similar to that of \emph{persistent homology}, which is a well-developed and the most widely-used tool in TDA; see \cite{curry} for complete details.

In this paper, we propose a topological approach to community detection based on sheaves.  Specifically, we show that sheaves may be used to rigorously model the problem of community detection on a network and propose three novel, sheaf-based community detection algorithms. We test and compare their performance with numerical experiments on a real-world benchmarking dataset and show that we are able to attain near optimal community detection results.

\paragraph{Related Work.}  Particularly relevant to our work, sheaves have been previously used to model various dynamics of opinions over social networks, including bounded confidence models, stubbornness, and the formation of lies \cite{hansenghrist}. In a similar spirit, propagation of gossip has also been theoretically modeled using sheaves \cite{ghristriess}.  However, it is important to note that no implementations nor applications to real data exist of these sheaf social network models. In a non-topological setting, the formation of opinion clusters has been investigated in the bounded confidence model \cite{hk}; another contrast to our work is that they use dynamics with discrete time steps and sharp confidence bounds.
Very recently, the only other TDA approach to community detection that we are aware of was proposed \cite{schindler2023persistent}, however, it is based on persistent homology, rather than sheaves, as in our work.

\section{Background: Sheaves and Social Networks}
\label{sec:background}

In this section, we define sheaves and present sheaf-theoretic notions in the context of social networks that will be used in our proposed community detection algorithms.

\subsection{Sheaves and Sheaf Cohomology}

Sheaves assign data to open subsets of topological spaces $X$ in a consistent manner. 
This compatibility is what enables them to relate \emph{local} (i.e., the pure data) and \emph{global} (i.e., having an assignment of compatible data) properties of $X$.
\emph{Cellular} sheaves are a special case where $X$ is a cellular complex (which is a generalization of a simplicial complex) and the data are vector spaces. In this paper, we will further restrict to the case that the cellular complex is a graph $G=(V,E)$; we consider undirected graphs without loops or multiple edges.  We write edges as ordered tuples $e=(v_1,v_2)$, ordered arbitrarily.  Under these restrictions to graphs, we now outline how important concepts from graph theory generalize to the topological setting of sheaves which will be relevant to our work further on.

The graph $G$ becomes a topological space when exploiting the fact that the \emph{subface relation} allows $G$ to be viewed as a partially ordered set and to be equipped with the Alexandrov topology. 
A sheaf on $G$ turns out to be uniquely determined by the following data \cite{curry}.
\begin{definition}
    A \emph{(cellular) sheaf} $\Ff$ assigns finite-dimensional real vector spaces, called \emph{stalks}, $\Ff(v)$ to each vertex and $\Ff(e)$ to each edge of a graph, and a linear map $\Ff_{v\subset e} :\Ff(v) \rightarrow \Ff(e)$ (called a \emph{restriction}) to each incidence $v\subset e$.
\end{definition}

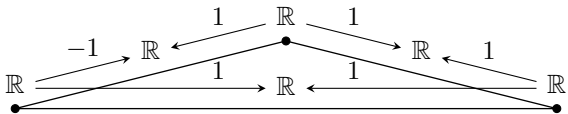
\begin{figure}[h]
    \begin{center} 
        \begin{tikzpicture}[x=0.9cm,y=0.9cm]
	\clip(0.5,0.5) rectangle (9.5,3);
    \draw [line width=.5pt] (1,1)-- (5,2);
    \draw [line width=.5pt] (5,2)-- (9,1);
    \draw [line width=.5pt] (1,1)-- (9,1);
    \draw [-stealth] (1.3,1.4) -- (2.7,1.75);
    \draw [stealth-] (3.3,1.9) -- (4.7,2.25);
    \draw [-stealth] (5.3,2.25) -- (6.7,1.9);
    \draw [stealth-] (7.3,1.75) -- (8.7,1.4);
    \draw [-stealth] (1.3,1.3) -- (4.7,1.3);
    \draw [stealth-] (5.3,1.3) -- (8.7,1.3);
    \draw[color=black, anchor=south] (2,1.6) node {$-1$};
    \draw[color=black, anchor=south] (4,2.1) node {$1$};
    \draw[color=black, anchor=south] (6,2.1) node {$1$};
    \draw[color=black, anchor=south] (8,1.6) node {$1$};
    \draw[color=black, anchor=south] (4,1.3) node {$1$};
    \draw[color=black, anchor=south] (6,1.3) node {$1$};
	\draw [fill=black] (1,1) circle (1.5pt);
    \draw[color=black, anchor=south] (1,1.1) node {$\IR$};
    \draw[color=black, anchor=south] (5,2.1) node {$\IR$};
    \draw[color=black, anchor=south] (9,1.1) node {$\IR$};
    \draw[color=black, anchor=south] (3,1.6) node {$\IR$};
    \draw[color=black, anchor=south] (7,1.6) node {$\IR$};
    \draw[color=black, anchor=south] (5,1.07) node {$\IR$};
 	\draw [fill=black] (9,1) circle (1.5pt);
	\draw [fill=black] (5,2) circle (1.5pt);
    \begin{scriptsize}
    \end{scriptsize}
\end{tikzpicture}\\ 
    \caption{Example of a cellular sheaf}
    \label{fig:noglobal} 
    \end{center}
\end{figure}

\Cref{fig:noglobal} illustrates a cellular sheaf. Other examples that can be defined for any graph $G$ are the \textit{constant sheaves} $\underline \IR^n$ for $n \in \IN$, where all stalks are $\IR^n$ and all restrictions the identity map.

Algebraic topology is a field of pure mathematics that uses abstract algebra to study topological spaces; specifically, it defines algebraic ways of counting properties of topological spaces that are left unaltered under continuous deformations of the topological space, such as stretching or compressing. Such properties are referred to as \emph{invariants}; 
\emph{cohomology groups} are examples of such invariants.
Sheaves over topological spaces give rise to \emph{sheaf cohomology groups}. In our setting, these are obtained from a collection of vector spaces and maps between them, called \emph{the cochain complex} \cite[Theorem 1.4.2]{she} 

$$
0 \longrightarrow \ C^0(G,\Ff) \xrightarrow[\hspace{.7cm}]{\delta} C^{1}(G,\Ff) \longrightarrow 0 \longrightarrow \cdots
$$
with cochain groups
$$
C^0(G,\Ff):= \bigoplus_{v\in V} \Ff(v), \qquad C^1(G,\Ff):=\bigoplus_{e\in E} \Ff(e).
$$
Here, the linear \emph{coboundary map} $\delta: C^0(G,\Ff) \rightarrow C^1(G,\Ff)$ is defined by acting linearly on stalks. On $\Ff(v)$, it acts according to
$$
    x_v \mapsto \sum_{e=(v_k,v)} \Ff_{v\subset e} \, x_v - \sum_{e=(v,v_l)} \Ff_{v\subset e} \, x_v.
$$
The sheaf cohomology groups of interest are then
$$
    H^0(G,\Ff)=\ker(\delta) \qquad H^1(G,\Ff)=C^1(G,\Ff) \big/ \text{im}(\delta) = \text{coker}(\delta).
$$
\begin{example}
    The graph in \Cref{fig:noglobal} has trivial cohomology with respect to the illustrated sheaf, even though the classical (e.g., simplicial) cohomology is nontrivial. This shows that these two cohomology theories need not agree.
\end{example}
Sheaf cohomology is in fact a generalization of classical cohomology, because for the constant sheaf $\underline \IR^1$, the sheaf cohomology yields precisely the cellular cohomology of $G$. In this case, $\delta=B^\top$ is the transpose of the \textit{signed incidence matrix} $B$ which is defined by
\begin{align*}
    B_{v,e}=\begin{cases} 1 & \exists \, w \in V:\, e=(v,w),\\
    -1 & \exists \, w \in V:\, e=(w,v),\\
    0 & \text{otherwise.}
    \end{cases}
\end{align*}
Recall that the classical \textit{graph Laplacian} can be obtained as $L=BB^\top$.  This notion may be generalized to obtain the following definition.
\begin{definition}
    For a cellular sheaf $\Ff$ over $G$, the \emph{sheaf Laplacian} is $L_\Ff:=\delta^\top \delta$.
\end{definition}
It can be shown that neither sheaf cohomology nor the sheaf Laplacian depend on the initially chosen orientation.

\subsection{Discourse Sheaves and Opinion Dynamics}
\label{sec:ds_od}

By considering social networks $G$ where persons are modeled by vertices and  acquaintanceship is modeled by connections, the distributions of opinions may be expressed by a \emph{discourse sheaf} $\Ff$ \citep{hansenghrist}; see \Cref{fig:discshf} for an illustration.

\begin{figure}[h]
\centering
    \begin{tikzpicture}[x=1cm,y=1cm]
	\clip(0,0.5) rectangle (8,2.5);
    \draw [line width=.5pt] (1,1)-- (7,1);
    \draw [line width=.5pt] (1,1)-- (0,1.2);
    \draw [line width=.5pt] (1,1)-- (0,0.5);
    \draw [line width=.5pt] (8,0.6)-- (7,1);
    \draw [line width=.5pt] (8,1.2)-- (7,1);
    \draw [-stealth,thick] (1.9,1.8) -- (3,1.8);
    \draw [stealth-,thick] (5,1.8) -- (6.1,1.8);
	\draw [fill=black] (1,1) circle (1.5pt);
    \node[draw, text width=1.2cm,align=center] at (1,1.8) {\small opinion space};
    \node[draw, text width=1.2cm,align=center] at (7,1.8) {\small opinion space};
    \node[draw, text width=1.5cm,align=center] at (4,1.8) {\small discourse space};
	\draw [fill=black] (7,1) circle (1.5pt);
\end{tikzpicture}
\caption{Sketch of a discourse sheaf}
\label{fig:discshf}
\end{figure}
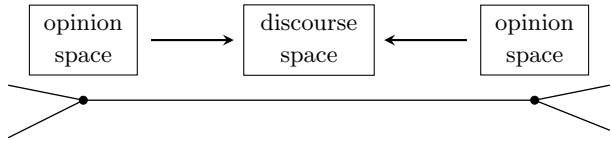

Each person $v$ is assigned an \emph{opinion space} $\IR^{n_v}=\Ff(v)$ and \emph{opinion} $x_v \in \Ff(v)$. A basis of $\Ff(v)$ can be seen as collection of basic topics that $v$ cares about and the component of $x_v$ in a basis direction expresses the opinion about that topic (e.g., how much $v$ supports a certain politician).

Edges stand for discourse about topics in $\IR^{n_e}=\Ff(e)$. There need not be a relation between bases of different stalks. However, each person $v$ projects their opinion of the discussed topics on $e$ via $\Ff_{v\subset e}$. There is \emph{consensus} along $e=(u,v)$ if $\Ff_{v\subset e} \, x_v = \Ff_{u \subset e } \, x_u$.

In this framework, several models have been proposed to describe how opinions expressed by such a sheaf evolve over time \cite{hansenghrist}. 
The basic model assumes that everyone changes their opinion in order to minimize the difference to the average opinion of their friends, i.e.,
$$
    \frac{d}{d t} x_v(t) = \sum_{v \mathrel{\overset{\makebox[0pt]{\mbox{\normalfont\small\sffamily $e$}}}{\sim}} u} \Ff_{v\subset e}^\top (\Ff_{u\subset e} \, x_u - \Ff_{v \subset e} \, x_v) 
$$
Here $v \mathrel{\overset{\makebox[0pt]{\mbox{\normalfont\small\sffamily $e$}}}{\sim}} u$ denotes that $v$ and $u$ share a common edge $e$. Combining all opinions to one vector $x$, this can be written as $\displaystyle \frac{d}{d t} x(t) = - L_\Ff \, x$. Solutions converge exponentially to consensus on all edges (\cite{hansenghrist}, Theorem 4.1).

The more realistic \emph{bounded confidence model} assumes that the influence of a friend's opinion on $v$ decreases if their opinions differ too much. The decay is expressed by a monotonically decreasing ``bump'' function $\phi:[0,\infty) \rightarrow [0,1]$ that vanishes precisely whenever some threshold $D$ is surpassed, as shown in \Cref{fig:phis}.
The modified dynamics are described by
\begin{equation}
    \frac{d}{d t} x_v(t) = \sum_{v \mathrel{\overset{\makebox[0pt]{\mbox{\normalfont\small\sffamily $e$}}}{\sim}} u} \phi \left(||\Ff_{u\subset e} \, x_u - \Ff_{v \subset e} \, x_v|| \right) \Ff_{v\subset e}^\top (\Ff_{u\subset e} \, x_u - \Ff_{v \subset e} \, x_v). \label{eq:bdcom}
\end{equation}
Configurations are stable if on each edge the difference is zero or at least $D$. In such a stable configuration, disregarding all edges without consensus and considering the connected components of the remaining graph provides a partition of the vertex set $V$ (persons in the social network). 

\section{Methods and Experimental Design}
\label{sec:alg}

In this section, we outline our proposed sheaf-theoretic algorithms used for community detection. 
We propose two algorithms that generate partitions of a given graph at random and one deterministic algorithm and evaluate their ability to detect communities. For the sake of comparability, we consider partitions of the full vertex set of a fixed graph (see \Cref{fig:karate} for an example of such a partition) and evaluate how well they represent the community structure of the graph. The common evaluation measure that we will use in this work is \emph{modularity}, which is computed in an assessment step to determine whether the division into communities is ``relatively good.''

\begin{figure}[h]
\centering
    \includegraphics[scale=0.4]{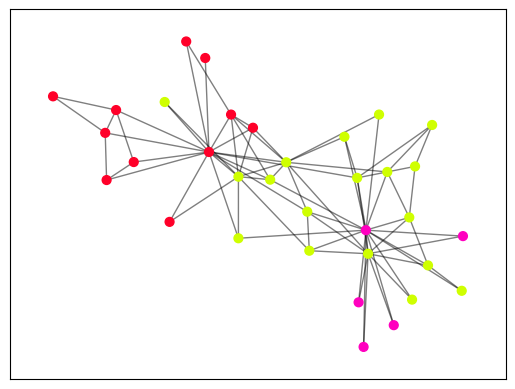}
    \caption{``Karate club'' graph partitioned into three communities.} 
    \label{fig:karate}
\end{figure}

\begin{definition}
\label{def:mod}
Let $G=(V,E)$ be a graph and $V$ be subdivided into $N$ subsets $\{V_c \mid c=1, \dots, N\}$, with $E_c$ being the set of edges between nodes from $V_c$. Then the modularity $Q$ is given by
\begin{equation}
\label{eq:modularity}
    Q:=\sum_{c=1}^N \left[\frac{|E_c|}{|E|} - \left(\frac{\sum_{v \in V_c} \deg (v)}{2|E|}\right)^2\right] \in [-1,1].
\end{equation}
\end{definition}
Intuitively, the quantity $Q$ captures how many more edges than the average are within the subsets $V_c$
, so high modularity indicates a good partitioning.

\subsection{Detecting Communities with Constant Sheaves}
Our first algorithm considers constant sheaves and models dynamics after the sheaf-theoretic bounded confidence model previously described in \Cref{sec:ds_od}.

\begin{algorithm}
\caption{Community Detection with Constant Sheaves $\underline \IR^n$}\label{alg:const}
\KwInput{dimension parameter $n$, parameter for diameter $d$}
Initialize by picking a random opinion vector for each vertex stalk uniformly from $B(0,\frac d 2)$\\
Evolve the system using \cref{eq:bdcom} until for no edge $e=(u,v)$, the difference $\|x_u - x_v\| \in (0.0033,1)$; abort if that does not happen within 1000 time units\\
Obtain primary partition by grouping neighbors $w,w'$ together if and only if $\|x_w-x_{w'}\| \le 0.0033$\\
\If{we find a single-vertex community $v$}
{
add the vertex to that adjacent community $C$, which has maximal $2|E| \, k_C - \deg(v) \cdot \sum_{w \in C} \deg(w)$, where $k_C$ is the number of neighbors of $v$ belonging to $C$
}
\KwOutput{The obtained partition}
\end{algorithm}

In Step 2 of \cref{alg:const}, we have included a stopping criterion for time efficiency: If the evolution does not converge after 1000 time units, the calculation is aborted. In \cref{sec:convergence} we explain why this is necessary. We keep track of the number of abortions and take them into account in our reported results below, when providing uncertainties.

We impose the positive value of 0.0033 for the computed difference in Steps 2 and 3 as another stopping criterion and trade-off between precision and time-efficiency.

Step 4 resolves single-vertex clusters by adding the vertex $v$ to the cluster of that neighbor of $v$ which is optimal in the sense of modularity. \Cref{prop:mod} justifies this choice by showing that each of these moves increases modularity.

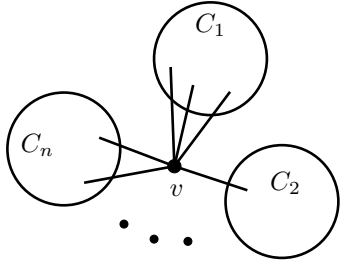
\begin{figure}[h]
\centering
\begin{tikzpicture}[x=0.5cm,y=0.5cm]
\clip(-3.1,-3) rectangle (6,5);
\draw [line width=1pt] (2.5,3.3) circle (.75cm);
\draw [line width=1pt] (4.4,-0.5) circle (.75cm);
\draw [line width=1pt] (-1.4,0.9) circle (.75cm);
\draw [line width=1pt] (3.48,-0.2)-- (1.54,0.44);
\draw [line width=1pt] (1.54,0.44)-- (3.02,2.42);
\draw [line width=1pt] (2.04,2.6)-- (1.54,0.44);
\draw [line width=1pt] (1.44,3.08)-- (1.54,0.44);
\draw [line width=1pt] (-0.46,1.2)-- (1.54,0.44);
\draw [line width=1pt] (-0.86,0)-- (1.54,0.44);
\draw [fill=black] (1.54,0.44) circle (2.5pt);
\draw[color=black] (1.6,-0.2) node {$v$};
\draw[color=black] (2.5,4.2) node {$C_1$};
\draw[color=black] (4.5,0) node {$C_2$};
\draw[color=black] (-2.1,1) node {$C_n$};
\draw [fill=black] (0.2,-1.1) circle (1.5pt);
\draw [fill=black] (1,-1.5) circle (1.5pt);
\draw [fill=black] (1.9,-1.56) circle (1.5pt);
\end{tikzpicture}
\caption{Sketch of a possible situation in the proof of \Cref{prop:mod} with $k_1=3,\ k_2=1,\ k_n=2$}
\label{fig:proof}
\end{figure}

\begin{prop}
Single vertex clusters can be removed in a way that increases modularity. The removal of single vertex clusters implemented in \cref{alg:const} is optimal in the sense of modularity. \label{prop:mod}
\end{prop}

\begin{proof}
Consider a vertex $v$, the clusters $C_1,\dots,C_n$ to which its neighbors belong, and let $k_i$ be the number of neighbors of $v$ that are in $C_i$, see \Cref{fig:proof}. Adding $v$ to $C_i$ increases the modularity of the given partition by
$$
\Delta_i:=\frac{k_i}{|E|}-\frac{2 \deg(v)\cdot \sum_{w \in C_i} \deg(w)}{4\, |E|^2}.
$$
Notice that the value that we maximize in Step 5 is $2|E| \, \Delta_i$ (for the cluster $C_i$ considered in Step 5).
We must show that at least for one $i$, $\Delta_i$ is positive. Observe that
\begin{align*}
4|E|^2 \sum_{i=1}^n \Delta_i&=\sum_{i=1}^n 4|E| k_i - 2 \deg(v) \sum_{i=1}^n \sum_{w \in C_i} \deg(w)\\
&> 4|E| \deg(v) -2 \deg(v) \cdot 2|E| > 0,
\end{align*}
because $\sum_{i=1}^n \sum_{w \in C_i} \deg(w)$ is bounded from above by the sum of all degrees of vertices that are not $v$ which is $2|E|-\deg(v)$. This completes the proof.
\end{proof}

\subsection{Convergence of \cref{alg:const}: Community Detection with Constant Sheaves} \label{sec:convergence}
The following example shows why we need the abortion criterion in Step 2 of \cref{alg:const}:
\begin{example}
Consider the graph shown in \cref{fig:noconvergence} and the constant sheaf $\underline \IR$ on that graph. Let $$\phi(x)=\begin{cases}
    1-x & 0 \le x \le 1 \\
    0  & x \ge 1
\end{cases}.$$ Initialize with values $a_0$ and $b_0$ such that $1+a_0>b_0>a_0$, making sure that none of the connections drawn in \cref{fig:noconvergence} is ignored. Applying \cref{eq:bdcom}, we check that all vertices with value $a(t)$ and $1+a(t)$ as well as all with $b(t)$ and $1+b(t)$ follow the same derivative. Therefore, the opinions in the network will remain in the pattern shown in \cref{fig:noconvergence} and the values $a(t)$ and $b(t)$ evolve according to
\begin{align*}
    \frac{da(t)}{dt}=(b(t)-a(t))(1+a(t)-b(t))=-\frac{db(t)}{dt}.
\end{align*}
That yields
\begin{align*}
    a(t)&=\frac{1}{2}\, \left(a_0+b_0-\frac{1}{e^{2t+c}+1} \right)\\
    b(t)&=\frac{1}{2}\, \left(a_0+b_0+\frac{1}{e^{2t+c}+1} \right)
\end{align*}
with $c=\ln\left(\frac{1}{b_0-a_0}-1\right)$. In particular, $1+a(t)-b(t)=1-\frac{1}{e^{2t+c}+1}<1$, but it converges to one. Thus, the difference over the edge in the middle of \cref{fig:noconvergence}, between $b(t)$ and $1+a(t)$ will always end up in the interval $(0.0033,1)$ and \cref{alg:const} will be aborted in Step 2.
\begin{figure}
    \centering
    \begin{tikzpicture}[x=1cm,y=1cm]
	\clip(-3.5,-1.3) rectangle (4.2,1.3);
    \draw [line width=.5pt] (-.4,0) -- (.2,0);
    \draw [line width=.5pt] (-1.6,0.3) -- (-2.4,0.7);
    \draw [line width=.5pt] (-1.6,-0.3) -- (-2.4,-0.7);
    \draw [line width=.5pt] (1.8,0.3) -- (2.6,0.7);
    \draw [line width=.5pt] (1.8,-0.3) -- (2.6,-0.7);
    \node[draw, text width=0.7cm,align=center] at (-1,0) {$b(t)$};
    \node[draw, text width=1.2cm,align=center] at (1,0) {$1+a(t)$};
    \node[draw, text width=0.7cm,align=center] at (-3,1) {$a(t)$};
    \node[draw, text width=0.7cm,align=center] at (-3,-1) {$a(t)$};
    \node[draw, text width=1.2cm,align=center] at (3.4,1) {$1+b(t)$};
    \node[draw, text width=1.2cm,align=center] at (3.4,-1) {$1+b(t)$};
    \end{tikzpicture}
    \caption{Example of a network in which \cref{alg:const} is aborted in Step 2}
    \label{fig:noconvergence}
\end{figure}
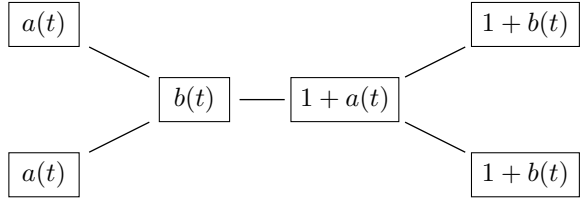
\end{example}
Even though there are configurations without convergence, in our experiments we observed a convergence within 1000 time steps for more than $95\%$ of starting configurations for any choice of parameters.

\subsection{Detecting Communities with a Non-Constant Sheaf}

We now consider a second community detection algorithm that uses the non-constant sheaf defined by setting $\Ff(v)=\IR^{\deg(v)}$ (we think of it as one copy of $\IR$ for every edge that uses $v$), $\Ff(e)=\IR$, and restriction maps being the projections onto the corresponding edges.

At the vertices, there is no interaction between the directions corresponding to different edges, so each direction can be treated separately. If for an edge $e=(v,w)$, the difference between the initial components $x_{v,e}$ and $x_{w,e}$ is less than $D$, edge $e$ will ``survive'' the evolution and otherwise not. Therefore, picking initial values at $v$ uniformly from $[-\frac d 2,\frac d 2]^{\deg(v)}$ amounts to keeping each edge with the same probability $p$ (which depends only on $d$) and ignoring it otherwise. We thus obtain \cref{alg:nonconst}.

\begin{algorithm}
\caption{Community Detection with a Non-Constant Sheaf}\label{alg:nonconst}
\KwInput{probability parameter $p$}
Obtain primary partition by grouping neighbors $w,w'$ together with probability $p$\\
Remove single-vertex clusters as in \cref{alg:const}, Step 4\\
\KwOutput{The obtained partition}
\end{algorithm}

\subsection{Deterministic Sheaf Community Detection}

A more general version of the bounded confidence model allows different functions $\phi_e$ for different edges. In the case of the non-constant sheaf, this means that the probability $p$ can depend on the edge and the local structure of the graph.

With the idea in mind that we want to retain edges if the vertices belong to the same community, it should be more likely for $e=(u,v)$ to be retained if the number of common neighbors of $u$ and $v$, denoted $N_{u,v}$, is not too small compared to the total number of neighbors of $u$ and $v$. In an extreme case, this likelihood is either zero or one, and we arrive at \cref{alg:dete}.

\begin{algorithm}
\caption{Deterministic Community Detection}\label{alg:dete}
\KwInput{two parameters $a \in [0,1]$ and $b \in \IR$}
Obtain primary partition by grouping the two ends $u,v$ of an edge $e$ together if $a \cdot (\deg(u)+\deg(v))< b+N_{u,v}$\\
Remove single-vertex clusters as in \cref{alg:const}, Step 4\\
\KwOutput{The obtained partition}
\end{algorithm}

\subsection{Experimental Setup}

We tested the performance of \cref{alg:const} by considering its dependence on $d$ for $n=1$ and the different bump functions $\phi$ shown in \Cref{fig:phis}, as well as for $n \in \{2,3,5,10\}$ and $\phi_1$. Note that for sake of comparability we fixed $D=1$ to be the threshold for ignoring.
On the interval $[0,1)$, the $\phi_i$ are given by:
\begin{align*}
    \phi_1(x)&=1-x\\
    \phi_2(x)&=1-x^2\\
    \phi_3(x)&=(1-x)^2\\
    \phi_4(x)&=1-x-\sin(2\,\pi\,x)/7
\end{align*}
\cref{alg:nonconst} was run for various values of $p$ and \cref{alg:dete} with different values of $a$ and $b$.

All tests were performed on Zachary's karate club graph $G$, shown in \Cref{fig:karate}, which represents a social network of a karate club studied by sociologist Wayne Zachary from 1970 to 1972 \citep{zachary1977information}. The network captures 34 members of the karate club and includes links between pairs of members who interacted outside the club. The maximal possible modularity of a partition of $G$ for this graph is $Q_\text{max}(G)\approx 0.42$ \cite{maxmod}.
All experiments were implemented in Python and the code is freely and publicly available at \url{https://github.com/ArneWolf/Bd_ confidence_communities}.

\begin{figure}[h]
\centering
    \begin{tikzpicture}
	\begin{axis}[
	height=5cm, width=.6\textwidth,
	only marks, xmin=0, xtick={0,0.2,...,1}, xmax=1.1,legend entries={$\ \phi_1$, $\ \phi_2$, $\ \phi_3$, $\ \phi_4$},
    legend style={at={(1,1)}, anchor=north east}
	]
	\addplot[domain=0:1, color=red, smooth, thick, mark size=0] {1-x};
 	\addplot[domain=0:1, color=blue, smooth, thick, mark size=0] {1-x^2};
	\addplot[domain=0:1, color=olive, smooth, thick, mark size=0] {(1-x)^2};
 	\addplot[domain=0:1, color=forest, smooth, thick, mark size=0] {1-x-sin(2*pi*deg(x))/7};
	\addplot[domain=1:1.1, color=black, smooth, thick, mark size=0] {0};
	\end{axis}
	\end{tikzpicture} 
    \caption{The different functions $\phi$}
    \label{fig:phis}
\end{figure}
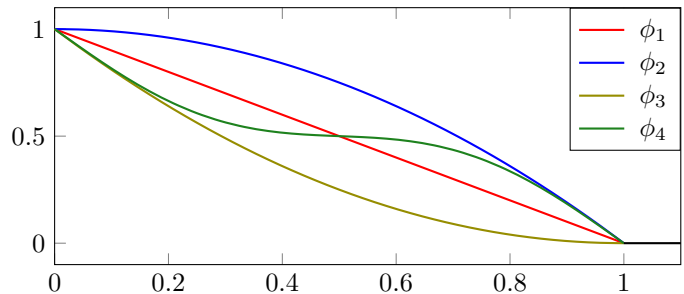

\section{Experimental Results} \label{sec:results}

We report the results of our numerical experiments for our three proposed sheaf-theoretic community detection algorithms from \Cref{sec:alg}. We compare the performance in terms of average number of clusters and modularity (\Cref{def:mod}).\\

\noindent
\textbf{\cref{alg:const}: The Constant Sheaves.} To justify our stopping criterion of 0.0033 in Step 2 of \cref{alg:const} and estimate the error caused, we compared the partitions to those obtained for running the evolution longer so that for every edge $\|x_u - x_v\| \not \in (0.001,1)$ instead. The more precise calculations took, on average, a factor of 100 longer. Out of 2930 runs we performed, the first stopping criterion applied 175 times and for the remaining 2755 runs, the two obtained partitions agreed. We thus concluded that 0.0033 is small enough to not cause a significant error.

Each parameter combination was simulated $N=1000$ times. \Cref{fig:constphis} shows the average values of $X$ (number of clusters, modularity) and an error bar with radius
$$
    X_{\text{error}}=\sigma_X+\frac{\# \, \text{early stoppings}}{N} \, \bar X,
$$
where $\bar X$ denotes the average and $\sigma$ the standard deviation. 

We immediately see that a larger width $d$ of the distribution from which the initial opinions are picked results in more clusters. Neither the number of clusters nor the modularity depends significantly on the choice of $\phi$.

\begin{figure}[h]
\centering
    \begin{tikzpicture}
	\begin{axis}[
	height=6cm, width=6.7cm,
	xlabel={$d$}, ylabel={}, ylabel style={at={(0,0.5)}},xmax=6.2,
	only marks,xlabel style={at={(0.5,0)}}
	]
	\addplot+ [color=red, mark=x, error bars/.cd, x dir=both, x explicit] table [x=p, y=pmax, col sep=comma] {pmaxplot.csv};
	\end{axis}
	\end{tikzpicture}
    \caption{Likelihood of most likely partition for different diameters $d$ of the initial distribution}
    \label{fig:pmax}
\end{figure}
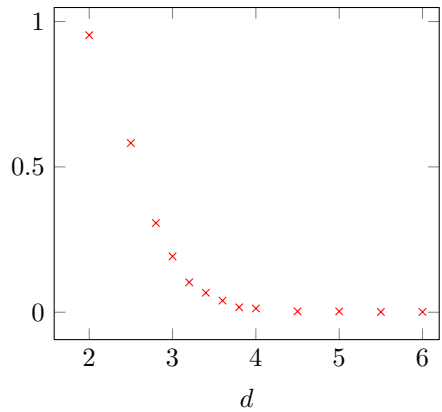

In the case of smaller $d$, however, $\phi_3$ tends to create more clusters than the others. Intuitively, this behavior can be explained by comparing $\phi_3$ to linearly decaying bump functions: it is more similar to a bump function that decays quicker than $\phi_1$, than to $\phi_1$. Quicker decay means a smaller threshold $D<1$. That in turn corresponds to having $D=1$ and a bigger $d$. Thus, the properties of $\phi_3$ tend to make $d$ appear larger. 

For large $d$, the influence of removing singleton clusters (compared to the primary partition) grows. This explains why the differences in the number of clusters as well as modularity for different choices of $\phi$ and $d$ become smaller. Looking at $p=0$ in \Cref{fig:nonconst} shows that resolving singleton clusters from a collection of only singleton clusters yields a modularity of about 0.191. This is one explanation for having higher modularites for larger $d$. Another explanation is that small $d$ are likely to result in one big cluster containing all the points. This trivial partition has modularity 0. \Cref{fig:pmax} shows the likelihood of the most likely partition for $\phi_1$ and $n=1$. For $d\le 4$, whenever there is a partition that is significantly the most likely one, this is the trivial partition. The same qualitative behavior is observed for other $\phi$ and $n$.

If $d$ is large enough, the number of clusters increases as $n$ increases. This is significant for $n=10$, but the tendency can be observed in \Cref{fig:constns} for smaller $n$ as well. This can be explained by larger $n$ corresponding to more topics of discussion, which in turn increases the chances that persons can disagree. The modularity for small $d$ seems to depend on $n$ non-monotonously: $n=10$ gives the highest values, followed by $n=1$ and $n=2$. If $d$ is large, the differences become insignificant.\\

\noindent
\textbf{\cref{alg:nonconst}: The Non-Constant Sheaf.} For \cref{alg:nonconst}, \Cref{fig:nonconst} shows the expected decrease of the number of clusters with increasing $p$. The modularity has a maximum of about 0.26 near $p=0.12$. For smaller $p \rightarrow 0$, the modularity becomes 0.191 and for large $p$ it goes to zero due to the dominance of the trivial partition.\\

\noindent
\textbf{\cref{alg:dete}: The Deterministic Algorithm.} Due to the affine linear condition in Step 1 of \cref{alg:dete}, the modularity and number of clusters change along affine lines. The deterministic algorithm reaches a maximal modularity of about $0.407$, which comes close to $Q_\text{max}=0.42$. However, the portion of the parameter space that reaches this value is small---the other algorithms had broader ranges of maxima for their modularities. The partition yielding $Q=0.407$ consists of four clusters, just as the partition obtaining $Q_\text{max}$.

\section{Discussion}
\label{sec:end}

In this work, we showed that sheaves are a viable algebraic topological tool to model the problem of community detection on networks.  We proposed three different algorithms, two of which had random initializations and were based on constant sheaves and a non-constant sheaf, and a third deterministic version that allows for different bump functions for different edges of the graph.  The deterministic sheaf-theoretic community detection algorithm, in particular, performed well in terms of modularity and achieved values near the maximal modularity value.  Ours is the first work to computationally implement sheaves on real-world social network data and the first use of sheaves in the problem of community detection on networks.  Moreover, our work provides a proof-of-concept for future work adapting the potential of cellular sheaves to studying complex systems captured by networks and simplicial complexes, in general.

Directions for future work involve combining various notions from the different algorithms to, for instance, allow bump functions in \cref{alg:const} or the probability $p$ in \cref{alg:nonconst} to depend on the edge, which may potentially improve performance.  An optimization procedure may also be proposed to find optimal parameters, and thus, optimal partitions for a network.

\begin{figure}
\begin{center}
\begin{tikzpicture}
     \begin{groupplot}[
         group style={
             group name=my plots,
             group size=1 by 2,
             xlabels at=edge bottom,
             xticklabels at=edge bottom,
             vertical sep=0pt
         },
         height=40cm, width=\textwidth
         ]
 	\nextgroupplot[
 	height=5cm, ylabel={average number of clusters},
 	only marks, ymin=0,legend entries={ $\ \phi_1$, $\ \phi_2$, $\ \phi_3$,$\ \phi_4$},
     legend style={at={(0.15,0.9)}, anchor=north east}
 	]
 	\addplot+ [color=red, mark=o, error bars/.cd, y dir=both, y explicit] table [x=r, y=num, y error=numerr,  col sep=comma] {phi1plot.csv};
  	\addplot+ [color=blue, mark=o, error bars/.cd, y dir=both, y explicit] table [x=r, y=num, y error=numerr, col sep=comma] {phi2plot.csv};
   	\addplot+ [color=olive, mark=o, error bars/.cd, y dir=both, y explicit] table [x=r, y=num, y error=numerr, col sep=comma] {phi3plot.csv};
    \addplot+ [color=forest, mark=o, error bars/.cd, y dir=both, y explicit] table [x=r, y=num, y error=numerr, col sep=comma] {phi4plot.csv};
    \nextgroupplot[height=6cm, ylabel={average modularity},xlabel={width $d$ of distribution},xlabel style={at={(0.5,0.03)}},ytick={-0.05,0,...,0.3},tick label style={/pgf/number format/fixed},
 	only marks]
    \addplot+ [color=red, mark=o, error bars/.cd, y dir=both, y explicit] table [x=r, y=qav, y error=qaverr,  col sep=comma] {phi1plot.csv};
  	\addplot+ [color=blue, mark=o, error bars/.cd, y dir=both, y explicit] table [x=r, y=qav, y error=qaverr, col sep=comma] {phi2plot.csv};
   	\addplot+ [color=olive, mark=o, error bars/.cd, y dir=both, y explicit] table [x=r, y=qav, y error=qaverr, col sep=comma] {phi3plot.csv};
    \addplot+ [color=forest, mark=o, error bars/.cd, y dir=both, y explicit] table [x=r, y=qav, y error=qaverr, col sep=comma] {phi4plot.csv};
\end{groupplot}
\end{tikzpicture} \vspace{-.7cm}
     \caption{Average number of clusters and modularity for $n=1$, different bump functions $\phi$ and different diameters of the initial opinion distributions}
     \label{fig:constphis}
     \end{center}\vspace{-.4cm}
\end{figure}
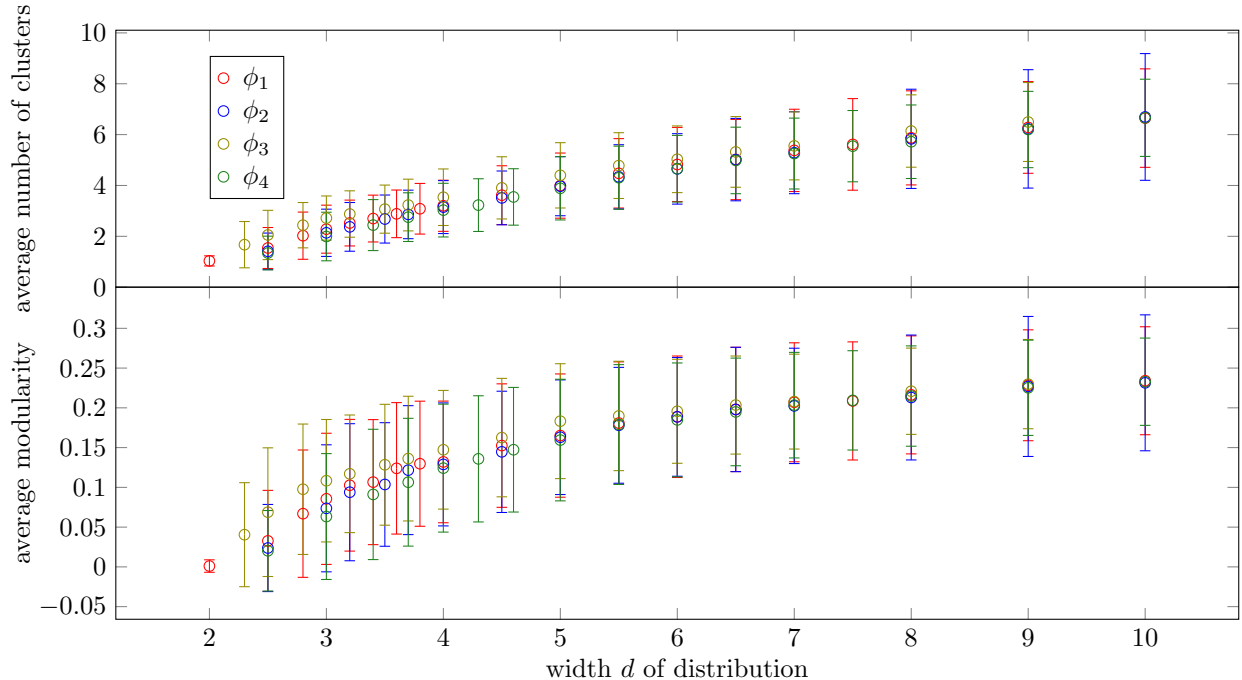
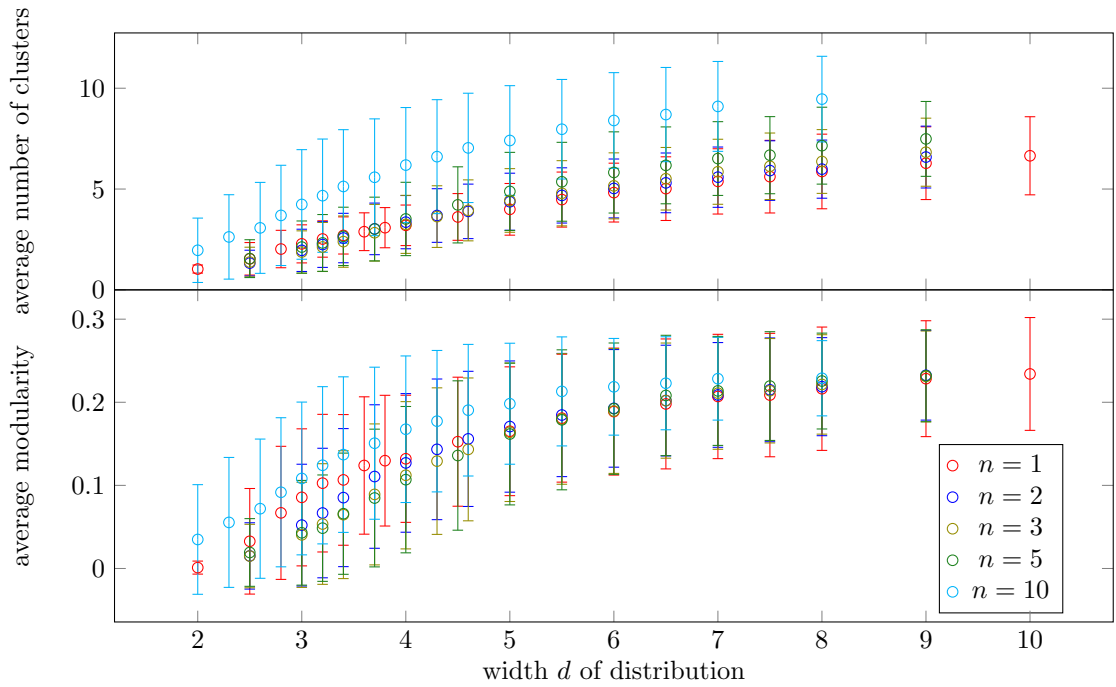
\begin{figure}
    \begin{center}
    \begin{tikzpicture}
    \begin{groupplot}[
        group style={
            group name=my plots,
            group size=1 by 2,
            xlabels at=edge bottom,
            xticklabels at=edge bottom,
            vertical sep=0pt
        },
        height=30cm, width=0.9\textwidth
        ]
	\nextgroupplot[
	height=5cm, ylabel={average number of clusters}, 
	only marks, ymin=0,legend entries={ $\ n=1$, $\ n=2$, $\ n=3$,$\ n=5$, $\ n=10$},
    legend style={at={(0.95,-0.6)}, anchor=north east}
	]
	\addplot+ [color=red, mark=o, error bars/.cd, y dir=both, y explicit] table [x=r, y=num, y error=numerr,  col sep=comma] {phi1plot.csv};
 	\addplot+ [color=blue, mark=o, error bars/.cd, y dir=both, y explicit] table [x=r, y=num, y error=numerr, col sep=comma] {const2plot.csv};
  	\addplot+ [color=olive, mark=o, error bars/.cd, y dir=both, y explicit] table [x=r, y=num, y error=numerr, col sep=comma] {const3plot.csv};
     \addplot+ [color=forest, mark=o, error bars/.cd, y dir=both, y explicit] table [x=r, y=num, y error=numerr, col sep=comma] {const5plot.csv};
    \addplot+ [color=skyblue, mark=o, error bars/.cd, y dir=both, y explicit] table [x=r, y=num, y error=numerr, col sep=comma] {const10plot.csv};
    \nextgroupplot[height=6cm, ylabel={average modularity},xlabel={width $d$ of distribution},xlabel style={at={(0.5,0.03)}},tick label style={/pgf/number format/fixed},
	only marks]
    \addplot+ [color=red, mark=o, error bars/.cd, y dir=both, y explicit] table [x=r, y=qav, y error=qaverr,  col sep=comma] {phi1plot.csv};
 	\addplot+ [color=blue, mark=o, error bars/.cd, y dir=both, y explicit] table [x=r, y=qav, y error=qaverr, col sep=comma] {const2plot.csv};
  	\addplot+ [color=olive, mark=o, error bars/.cd, y dir=both, y explicit] table [x=r, y=qav, y error=qaverr, col sep=comma] {const3plot.csv};
     \addplot+ [color=forest, mark=o, error bars/.cd, y dir=both, y explicit] table [x=r, y=qav, y error=qaverr, col sep=comma] {const5plot.csv};
    \addplot+ [color=skyblue, mark=o, error bars/.cd, y dir=both, y explicit] table [x=r, y=qav, y error=qaverr, col sep=comma] {const10plot.csv};
\end{groupplot}
	\end{tikzpicture} \vspace{-.3cm}
    \caption{Average number of clusters and modularity for different dimensions $n$, bump function $\phi_1$ and different diameters of the initial opinion distributions} 
    \label{fig:constns}
    \end{center}\vspace{-1cm}
\end{figure}

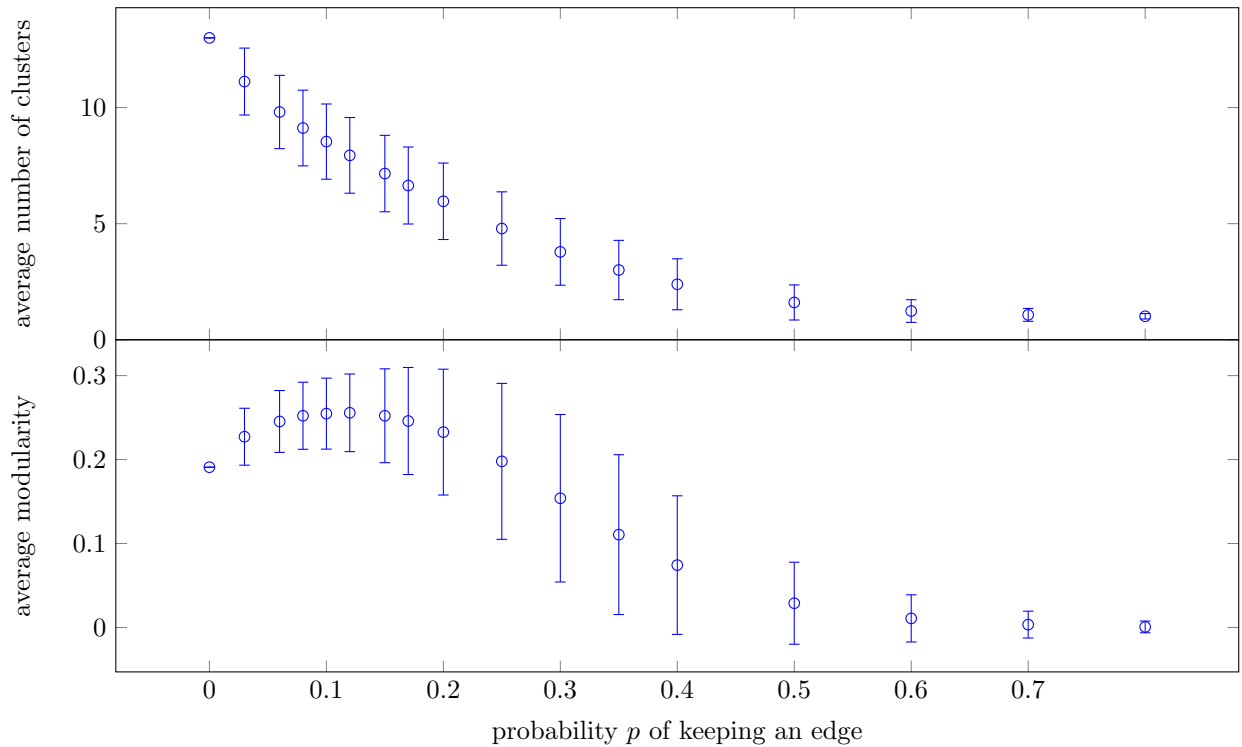
\begin{figure}
    \begin{center}
    \begin{tikzpicture}
    \begin{groupplot}[
        group style={
            group name=my plots,
            group size=1 by 2,
            xlabels at=edge bottom,
            xticklabels at=edge bottom,
            vertical sep=0pt
        },
        height=28cm, width=\textwidth
        ]
	\nextgroupplot[
	height=6cm, ylabel={average number of clusters},
	only marks, ymin=0
	]
	\addplot+ [color=blue, mark=o, error bars/.cd, y dir=both, y explicit] table [x=p, y=num, y error=numerr,  col sep=semicolon] {nonconstplot.csv};
    \nextgroupplot[height=6cm, ylabel={average modularity}, xlabel={probability $p$ of keeping an edge},xlabel style={at={(0.5,0)}},tick label style={/pgf/number format/fixed},xtick={0,0.1,...,0.7},
	only marks]
 	\addplot+ [color=blue, mark=o, error bars/.cd, y dir=both, y explicit] table [x=p, y=qav, y error=qaverr, col sep=semicolon] {nonconstplot.csv};
\end{groupplot}
\end{tikzpicture}
    \caption{Average number of clusters and modularity for the non-constant sheaf algorithm for different probabilities $p$ of keeping an edge}
    \label{fig:nonconst}
    \end{center}
\end{figure}

\begin{figure}
    \begin{center}
	\includegraphics{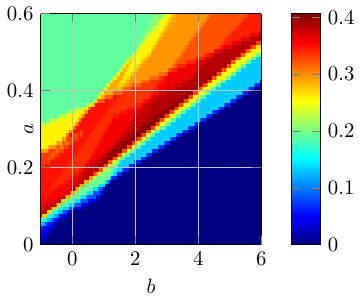} \hspace{1cm}
    \includegraphics{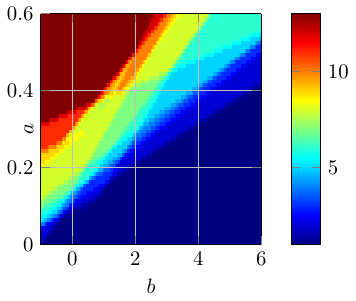}
    \caption{Modularities (left) and numbers of clusters (right) for several parameter combinations for the deterministic \cref{alg:dete}}
    \label{fig:detemod}
    \end{center}
\end{figure}

\section*{Acknowledgments}

A.W.~is funded by a London School of Geometry and Number Theory--Imperial College London PhD studentship, which is supported by the Engineering and Physical Sciences Research Council [EP/S021590/1].

\bibliographystyle{authordate3}
\bibliography{sheafCD_ref}

\end{document}